\documentclass[11pt]{article}

\usepackage{fullpage}

\usepackage{amsthm,amsfonts,amssymb,amsmath,amsxtra}
\usepackage[dvipsnames]{xcolor}
\usepackage{pgfplots}
\usepackage{xspace,enumerate}
\usepackage[utf8]{inputenc}
\usepackage{thmtools}
\usepackage{thm-restate}
\usepackage{authblk}
\definecolor{TUMBlau}{HTML}{0065bd}
\usepackage[hypertexnames=false,colorlinks=true,urlcolor=TUMBlau,citecolor=TUMBlau,linkcolor=TUMBlau]{hyperref}
\usepackage[noadjust]{cite}
\usepackage{microtype}
\usepackage{todonotes}
\usepackage{thmtools}
\usepackage[capitalise]{cleveref}
\usepackage{algorithm}
\usepackage[noend]{algpseudocode}
\usepackage[pagewise]{lineno}
\usetikzlibrary{patterns}
\nolinenumbers

\theoremstyle{plain}
\newtheorem{theorem}{Theorem}
\newtheorem{lemma}[theorem]{Lemma} 
\newtheorem{corollary}[theorem]{Corollary}  
\newtheorem{proposition}[theorem]{Proposition}  
\newtheorem{definition}[theorem]{Definition}
\newtheorem{mainlemma}{Main Lemma}
\newtheorem*{loadlemma}{Load Lemma}

\newcommand{\BN}{\ensuremath{\mathbb{N}}\xspace}
\newcommand{\bE}{\ensuremath{\mathbf{E}}\xspace}
\newcommand{\bP}{\ensuremath{\mathbf{P}}\xspace}
\newcommand{\CJ}{\ensuremath{\mathcal{J}}\xspace}
\newcommand{\OPT}{{\mathrm{OPT}}}

\author[1]{Susanne Albers}
\author[2]{Maximilian Janke}
\affil[1]{Department of Computer Science, Technical University of Munich\\
\href{mailto:albers@in.tum.de}{albers@in.tum.de}}

\affil[2]{Department of Computer Science, Technical University of Munich\\
\href{mailto:maximilian.janke@in.tum.de}{janke@in.tum.de}}

\date{}

\title{Scheduling in the Random-Order Model\footnote{Work supported by the European Research Council, Grant Agreement No. 691672, project APEG.}}

\begin{document}

\maketitle

\thispagestyle{empty}

\begin{abstract}
Makespan minimization on identical machines is a fundamental problem in online scheduling. The goal is to
assign a sequence of jobs to $m$ identical parallel machines so as to minimize the maximum completion time of any job.
Already in the 1960s, Graham showed that {\em Greedy\/} is $(2-1/m)$-competitive~\cite{G}. The best deterministic
online algorithm currently known achieves a competitive ratio of 1.9201~\cite{FW}. No deterministic online
strategy can obtain a competitiveness smaller than 1.88~\cite{R}.

In this paper, we study online makespan minimization in the popular random-order model, where the jobs of
a given input arrive as a random permutation. It is known that {\em Greedy\/} does not attain a
competitive factor asymptotically smaller than~2 in this setting~\cite{OT}. We present the first improved
performance guarantees. Specifically, we develop a deterministic online algorithm that achieves a competitive
ratio of 1.8478. The result relies on a new analysis approach. We identify a set of properties that a random
permutation of the input jobs satisfies with high probability. Then we conduct a worst-case analysis
of our algorithm, for the respective class of permutations. The analysis implies that the stated competitiveness
holds not only in expectation but with high probability. Moreover, it provides mathematical evidence that
job sequences leading to higher performance ratios are extremely rare, pathological inputs. We complement 
the results by lower bounds, for the random-order model. We show that no deterministic online algorithm 
can achieve a competitive ratio smaller than 4/3. Moreover, no deterministic online algorithm can attain 
a competitiveness smaller than 3/2 with high probability.
\end{abstract}

\clearpage
\setcounter{page}{1}
\newpage
\section{Introduction}
We study one of the most basic scheduling problems. Consider a sequence of jobs ${\cal J}  = J_1, \ldots, J_n$ that
has to be assigned to $m$ identical parallel machines. Each job $J_t$ has an individual processing time $p_t$, $1\leq t\leq n$.
Preemption of jobs is not allowed. The goal is to minimize the {\em makespan\/}, i.e.\ the maximum completion time of any
job in the constructed schedule. Both the offline and online variants of this problem have been studied extensively,
see e.g.~\cite{BFKV,EOW,FW,G,HS,PST} and references therein.

We focus on the online setting, where jobs arrive one by one. Whenever a job $J_t$ is presented, its processing time
$p_t$ is revealed. The job has to be scheduled immediately on one of the machines without knowledge of any future jobs $J_s$,
with $s>t$. Given a job sequence ${\cal J}$, let $A({\cal J})$ denote the makespan of an online algorithm $A$ on ${\cal J}$.
Let ${\mathit OPT}({\cal J})$ be the optimum makespan. A deterministic online algorithm $A$ is {\em $c$-competitive\/} if
$A({\cal J}) \leq c\cdot {\mathit OPT\/}({\cal J})$ holds for all ${\cal J}$~\cite{ST}. The best competitive ratio that can
be achieved by deterministic online algorithms is in the range $[1.88,1.9201]$. No randomized online algorithm is known
that beats deterministic ones, for general $m$.

In this paper we investigate online makespan minimization in the random-order model. Here an input instance / job sequence
is chosen by an adversary. Then a random permutation of the input elements / jobs arrives. The random-order model was
considered by Dynkin~\cite{D} and Lindley~\cite{L} for the secretary problem. Over the last years the framework has received
quite some research interest and many further problems have been studied. These include generalized secretary
problems~\cite{BIK1,BIK2,FSZ,K2,L}, the knapsack problem~\cite{BIK1,K2}, bin packing~\cite{K}, facility location~\cite{M},
matching problems~\cite{GM,KMT,MQ}, packing LPs~\cite{KRTV} and convex optimization~\cite{GMM}.

We present an in-depth study of online makespan minimization in the random-order model. As a main contribution we
devise a new deterministic online algorithm that achieves a competitive ratio of~1.8478. After almost 20 years this 
is the first progress for the pure online setting, where an algorithm does not resort to extra resources in handling 
a job sequence.

\smallskip

{\bf Previous work:}
We review the most important results relevant to our work and first address the standard setting where an online algorithm
must schedule an arbitrary, worst-case job sequence. Graham in 1966 showed that the famous {\em Greedy\/} algorithm, which
assigns each job to a least loaded machine, is $(2-{1\over m})$-competitive. Using new deterministic strategies the competitiveness
was improved in a series of papers. Galambos and Woeginger~\cite{GW} gave an algorithm with a competitive ratio of
$(2-{1\over m}-\epsilon_m)$, where $\epsilon_m$ tends to 0 as $m\rightarrow \infty$. Bartal et al.~\cite{BFKV}
devised a 1.986-competitive algorithm. The bound was improved to 1.945~\cite{KPT} and 1.923~\cite{A}.
Fleischer and Wahl~\cite{FW} presented an algorithm that attains a competitive ratio of 1.9201 as $m\rightarrow \infty$.
Chen et al.~\cite{CYZ} gave an algorithm whose competitiveness is at most $1+\varepsilon$ times the best possible factor,
but no explicit bound was provided. Lower bounds on the competitive ratio of deterministic online algorithms
were shown in~\cite{A,BKR,FKT,GRTW,R,RC}. For general $m$, the bound was raised from 1.707~\cite{FKT} to 1.837~\cite{BKR}
and 1.854~\cite{GRTW}. Rudin~\cite{R} showed that no deterministic strategy has a competitiveness smaller than 1.88.

For randomized online algorithms, there is a significant gap between the best known upper and
lower bounds. For $m=2$ machines, Bartal et al.~\cite{BFKV} presented an algorithm that achieves an
optimal competitive ratio of 4/3. To date, there exists no randomized algorithm whose competitiveness is smaller than the
deterministic lower bound, for general $m$. The best known lower bound on the performance of randomized online algorithms
tends to $e/(e-1)\approx 1.581$ as $m\rightarrow \infty$~\cite{CVW,Sg}.

Recent research on makespan minimization has examined settings where an online algorithm is given extra
resources when processing a job sequence. Specifically, an algorithm might have a buffer to reorder the
incoming job sequence~\cite{EOW,KKST} or is allowed to migrate jobs~\cite{SSS}. Alternatively, an
algorithm has information on the job sequence~\cite{CKK,D2,KK,KKST}, e.g.\ it might know the total
processing time of the jobs or even the optimum makespan.

In the random-order model only one result is known for makespan minimization on identical machines. Osborn
and Torng~\cite{OT} showed that {\em Greedy\/} does not achieve a competitive ratio smaller than~2 as
$m\rightarrow \infty$. Recently Molinaro~\cite{M2} studied online load balancing with the objective to
minimize the $l_p$-norm of the machine loads. He considers a general scenario with machine-dependent job
processing times, which are bounded by $1$. For makespan minimization he presents an algorithm that, in the worst case, is
$O(\log m/\varepsilon)$-competitive and, in the random-order model, has an expected makespan of
$(1+\varepsilon) {\mathit OPT}({\cal J}) + O(\log m/\varepsilon)$, for any $\varepsilon \in (0,1]$. G\"obel et al.\ \cite{gobel2015online} consider a scheduling problem on one machine where the goal is to minimize the average weighted completion time of all jobs. Under random-order arrival, their competitive ratio is logarithmic in $n$, the number of jobs, for the general problem and constant if all jobs have processing time $1$.

\smallskip

{\bf Our contribution:} We investigate online makespan minimization in the random-order model, a sensible and
widely adopted input model to study algorithms beyond the worst case. Specifically, we develop a new deterministic
algorithm that achieves a competitive ratio of $1.8478$ as $m\rightarrow \infty$. This is the first improved
performance guarantee in the random-order model. The competitiveness is substantially below the best known ratio
of~1.9201 in the worst-case setting and also below the corresponding lower bound of 1.88 in that framework.

A new feature of our algorithm is that it schedules an incoming job on one of three candidate machines in order
to maintain a certain load profile. The best strategies in the worst-case setting use
two possible machines, and it is not clear how to take advantage of additional machines in that framework.
The choice of our third, extra machine is quite flexible: An incoming job is placed either on a least loaded,
a heavily loaded or -- as a new option -- on an intermediate machine. The latter one is the $(h+1)$-st least loaded
machine, where $h$ may be any integer with $h \in \omega(1)$ and $h\in o(\sqrt{m})$.

When assigning a job to a machine different from the least loaded one, an algorithm has to ensure that the
resulting makespan does not exceed $c$ times the optimum makespan, for the targeted competitive ratio $c$.
All previous strategies in the literature lower bound the optimum makespan by the current average load on the
machines. Our new algorithm works with a refined lower bound that incorporates the processing times of the largest
jobs seen so far. The lower bound is obvious but has not been employed by previous algorithms.

The analysis of our algorithm proceeds in two steps. First we define a class of {\em stable job sequences\/}.
These are sequences that reveal information on the largest jobs as processing volume is scheduled.
More precisely, once a certain fraction of the total processing volume $\sum_{t=1}^n p_t$ has arrived, one
has a good estimate on the $h$-th largest job and has encountered a certain number of the $m+1$ largest jobs
in the input. The exact parameters have to be chosen carefully.

We prove that with high probability, a random permutation of a given input of jobs is stable. We then conduct
a worst-case analysis of our algorithm on stable sequences. Using their properties, we show that if the algorithm generates a
flat schedule, like {\em Greedy\/}, and can be hurt by a huge job, then the input must contain many large jobs so that
the optimum makespan is also high. A new ingredient in the worst-case analysis is the processing time of the $h$-th largest
job in the input. We will relate it to machine load in the schedule and to the processing time of the $(m+1)$-st largest job; 
twice the latter value is a lower bound on the optimum makespan.

The analysis implies that the competitive ratio of 1.8478 holds with high probability. Input sequences leading to
higher performance ratios are extremely rare. We believe that our analysis approach might be fruitful in the
study of other problems in the random-order model: Identify properties that a random permutation of the input
elements satisfies with high probability. Then perform a worst-case analysis.

Finally in this paper we devise lower bounds for the random-order model. We prove that no deterministic online 
algorithm achieves a competitive ratio smaller than $4/3$. Moreover, if a deterministic online algorithm is
$c$-competitive with high probability, then $c\geq 3/2$.

\section{Strong competitiveness in the random-order model}\label{sec.strongcc}

We define competitiveness in the random-order model and introduce a stronger measure of competitiveness
that implies high-probability bounds. Recall that traditionally a deterministic online algorithm $A$ is
$c$-competitive if $A({\cal J})\leq c \cdot OPT({\cal J})$ holds for all job sequences
${\cal J} = J_1, \ldots, J_n$. We will refer to this worst-case model also as the {\em adversarial model}.

In the {\em random-order model} a job sequence ${\cal J} = J_1, \ldots, J_n$ is given, which may be specified
by an adversary. (Alternatively, a set of jobs could be specified.) Then a random permutation of the jobs arrives.
We define the expected cost / makespan of a deterministic online algorithm. Let $S_n$ be the permutation group
of the integers from $1$ to $n$, which we consider a probability space under the uniform distribution, i.e.\ each
permutation in $S_n$ is chosen with probability $1/n!$. Given $\sigma\in S_n$, let ${\cal J}^{\sigma} =
J_{\sigma(1)},\ldots, J_{\sigma(n)}$ be the \emph{job sequence permuted by $\sigma$}. The expected makespan of $A$ on
${\cal J}$ in the random-order model is
$A^\textrm{rom}({\cal J}) = \bE_{\sigma\sim S_n}[A({\cal J}^\sigma)] = {1\over n!} \sum_{\sigma \in S_n} A({\cal J}^\sigma)$.
The algorithm $A$ is {\em $c$-competitive in the random-order model\/} if
$A^\textrm{rom}({\cal J})\leq c \cdot OPT({\cal J})$ holds for all job sequences ${\cal J}$.

We next define the notion of a deterministic online algorithm $A$ being {\em nearly $c$-competitive\/}. The second condition in
the following definition requires that the probability of $A$ not meeting the desired performance ratio must be arbitrarily
small as $m$ grows and a random permutation of a given job sequence arrives. The subsequent Lemma~\ref{lem:comp} states
that a nearly $c$-competitive algorithm is $c$-competitive in the random-order model.

\begin{definition}\label{def:comp}
A deterministic online algorithm $A$ is called {\em nearly $c$-competitive\/} if the following two conditions hold.\\[-10pt]
\begin{itemize}
\item The algorithm $A$ achieves a constant competitive ratio in the adversarial model.\\[-10pt]
\item For every $\varepsilon >0$, there exists an $m(\varepsilon)$ such that for all machine numbers $m \geq m(\varepsilon)$
and all job sequences ${\cal J}$ there holds $\bP_{\sigma\sim S_n} [A({\cal J}^\sigma) \geq (c+\varepsilon) OPT({\cal J})] \leq \varepsilon$.
\end{itemize}
\end{definition}
\begin{lemma}\label{lem:comp}
If a deterministic online algorithm is nearly $c$-competitive, then it is $c$-competitive in the random-order model as $m\rightarrow \infty$.
\end{lemma}

\begin{proof}
Let $C$ be the constant such that $A$ is $C$-competitive in the adversarial model. We may assume that $C>c$.
Given $0<\delta \leq C-c$, we show that there exists an $m(\delta)$ such that, for all $m\geq m(\delta)$, we have
$A^\textrm{rom}({\cal J})\leq (c+\delta) OPT({\cal J})$  for every job sequences ${\cal J}$. Let
$\varepsilon = \delta/(C-c+1)$. Since $A$ is nearly $c$-competitive, there exists an $m(\varepsilon)$ such
that, for all $m\geq m(\varepsilon)$ and all inputs ${\cal J}$, there holds
$P_{\varepsilon}({\cal J}) =\bP_{\sigma\sim S_n} [A({\cal J}^\sigma) \geq (c+\varepsilon) OPT({\cal J})]
\leq \varepsilon$. Set $m(\delta) = m(\varepsilon)$. We obtain
\begin{align*}
A^\textrm{rom}({\cal J})&\leq (1-P_{\varepsilon}({\cal J}))(c+\varepsilon)OPT({\cal J}) +
P_{\varepsilon}({\cal J})\cdot C \cdot OPT({\cal J})\\
&\leq  ((1-\varepsilon)(c+\varepsilon)+\varepsilon C) OPT({\cal J})\\
&\leq  (c+\varepsilon(C-c+1))OPT({\cal J})\\
&= (c+\delta) OPT({\cal J}). \qedhere
\end{align*}
\end{proof}

\section{Description of the new algorithm}\label{sec:alg}
The deficiency of {\em Greedy\/} is that it tends to generate a flat, balanced schedule in which all the machines
have approximately the same load. An incoming large job can then enforce a high makespan relative to the optimum one. 
It is thus crucial to try to avoid flat schedules and maintain steep schedules that exhibit a certain load imbalance 
among the machines.

However, in general, this is futile. Consider a sequence of $m$ identical jobs with a processing time of, say, $P_{m+1}$ (referring to the size of the $(m+1)$-st largest job
in an input). Any online algorithm that is better than $2$-competitive must schedule these $m$ jobs on separate machines, obtaining the flattest schedule 
possible. An incoming even larger job of processing time $p_\mathrm{max}$ will now enforce a makespan of $P_{m+1}+p_\mathrm{max}$. 
Observe that $\OPT\ge\max\{2P_{m+1},p_\mathrm{max}\}$ since there must be one machine containing two jobs. In particular $P_{m+1}+p_\mathrm{max}\le  1.5\OPT$.
Hence sensible online algorithms do not perform badly on this sequence.

This example summarizes the quintessential strategy of online algorithms that are good on all sequences: Ensure that in order to create a schedule that is very flat, 
i.e.\ such that all machines have high load $\lambda$, the adversary must present $m$ jobs that all are large relative to $\lambda$. In order to exploit this very flat schedule and cause a high makespan the adversary needs to follow up with yet another large job. But with these $m+1$ jobs, the optimum scheduler runs into the same problem as in the example: Of the $m+1$ large jobs, two have to be scheduled on the same machine. Thus the optimum makespan is high, compensating to the high makespan of the algorithm.

Effectively realizing the aforementioned strategy is highly non-trivial. In fact it is the central challenge in previous works on adversarial makespan minimization that improve upon \emph{Greedy} \cite{A,BFKV,FW,GW,KPT}. These works gave us clear notions of how to avoid flat schedules, which form the basis for our approaches. Instead of simply rehashing these ideas, we want to outline next how we profit from random-order arrival in particular. 

\subsection{How random-order arrival helps}
The first idea to profit from random-order arrival addresses the lower bound on $\OPT$ sophisticated online algorithms need. In the literature only the current average load has been considered, but under random-order arrival another bound comes to mind: The largest job seen so far. In order for an algorithm to perform badly, a large job needs to come close to the end of the sequence. Under random-order arrival, it is equally likely for such a job to arrive similarly close to the beginning of the sequence. In this case, the algorithm knows a better lower bound for $\OPT$. The main technical tool will be our \emph{Load Lemma}, which allows us to relate what a job sequence 
should reveal early from an analysis perspective to the actual fraction of jobs scheduled. This idea does not work for worst-case orders since they tend to order 
jobs by increasing processing times.

Recall that the general challenge of our later analysis will be to establish that there had to be $m$ large jobs once the schedule gets very flat. In classical analyses, which consider worst-case orders, these jobs appear with increasing density towards the end of the sequence. In random orders this is unlikely, which can be exploited by the algorithm.

The third idea improves upon the first idea. Suppose, that we were to modify our algorithm such that it could handle one very large job arriving close to the end of the sequence. In fact, assume that it could only perform badly when confronted with $h$ very large jobs. We can then disregard any sequence which contains fewer such jobs. Recall that the first idea requires one very large job to arrive sufficiently close to the beginning. Now, as $h$ grows, the probability of the latter event grows as well and approaches $1$. This will not only improve our competitive ratio tremendously, it also allows us to adhere to the stronger notion of nearly competitiveness introduced in \Cref{sec.strongcc}. Let us discuss how such a modification is possible: The first step is to design our algorithm in a way that it is reluctant to use the $h$ least loaded machines. Intuitively, if the algorithm tries to retain machines of small load it will require very large jobs to fill them. In order to force these filling jobs to actually be large enough, our algorithm needs to use a very high lower bound for $\OPT$. In fact, here it uses another lower bound for the optimum makespan, $2P_{m+1}^t$, twice the $(m+1)$-st largest job seen so far at time $t$. Common analysis techniques can only make predictions about $P_{m+1}^t$ at the very end of the sequence. It requires very subtle use of the random-order model to work around this.

\subsection{Formal definition}

Formally our algorithm ${\mathit ALG\/}$ is nearly $c$-competitive, where $c$ is the unique real root of the polynomial
$Q[x] = 4x^3-14x^2+16x -7$, i.e.
$$\textstyle c= {7 + \sqrt[3]{28 - 3\sqrt{87}} + \sqrt[3]{28 + 3\sqrt{87}} \over 6} < 1.8478.$$
Given ${\cal J}$, ${\mathit ALG\/}$ is presented with a job sequence/permutation ${\cal J}^\sigma =
J_{\sigma(1)}, \ldots, J_{\sigma(n)}$ that must be scheduled in this order.
Throughout the scheduling process ${\mathit ALG\/}$ always maintains a list of the machines sorted in
non-increasing order of current \emph{load}. At any time the load of a machine is the sum of the processing
times of the jobs already assigned to it. After ${\mathit ALG\/}$ has processed the first $t$ jobs
$J_{\sigma(1)},\ldots, J_{\sigma(t)}$, let $M_1^t, \ldots, M_m^t$ be any ordering of the $m$ machines according to
non-increasing load. More specifically, let $l_j^t$ denote the load of machine $M_j^t$. Then
$l_1^t\geq \ldots \geq l_m^t$ and $l_1^t$ is the makespan of the current schedule.

${\mathit ALG\/}$ places each incoming job $J_{\sigma(t)}$, $1\leq t \leq n$,  on one of three candidate machines. The choice of one
machine, having an intermediate load, is flexible. Let $h=h(m)$ be an integer with $h(m)\in \omega(1)$ and
$h(m)\in o(\sqrt{m})$. We could use e.g.\  $h(m) = \lfloor \sqrt[3]{m} \rfloor$ or
$h(m) = \lfloor \log m \rfloor$. Let $$i = \lceil (2c-3)m \rceil +h \approx 0.6956m.$$
${\mathit ALG\/}$ will assign the incoming job to the machine with the smallest load, the $(h+1)$-st smallest
load or the $i$-th largest load.

When scheduling a job on a machine that is different from the least loaded one, an algorithm has to ensure
that the resulting makespan does not exceed $c^*$ times the optimum makespan, where $c^*$ is the desired
competitiveness. All previous algorithms lower bound the optimum makespan by the current average machine load.
Algorithm ${\mathit ALG\/}$ works with a refined lower bound that incorporates the processing time of the
largest job and twice the processing time of the $(m+1)$-st largest job seen so far. These lower bounds on
the optimum makespan are immediate but have not been used in earlier strategies.

Formally, for $j=1,\ldots, m$, let $L_j^t$ be the \emph{average load} of the $m-j+1$ least loaded machines
$M_j^t, \ldots, M_m^t$, i.e.\ $L_j^t = {1\over m-j+1} \sum_{r=j}^m l_r^t$. We let
$L^t = L_1^t = {1\over m } \sum_{s=1}^t p_s$ be the average load of all the machines. For any
$j=1,\ldots, n$, let $P_j^t$ be the processing time of the $j$-th largest job among the first
$t$ jobs $J_{\sigma(1)}, \ldots, J_{\sigma(t)}$ in ${\cal J}^\sigma$. If $t<j$, we set $P_j^t = 0$. We let
$p^t_{\max} = P_1^t$ be the processing time of the largest job among the first $t$ jobs in ${\cal J}^\sigma$. 
Finally, let $L =L^n$, $P_j = P_j^n$ and $p_{\max} = p_{\max}^n$. 

The value $O^t = \max\{L^t, p_{\max}^t, 2P_{m+1}^t\}$ is a common lower bound on the optimum makespan
for the first $t$ jobs and hence $OPT({\cal J})$, see Proposition~\ref{prop:opt} in the next section.
Note that immediately before $J_{\sigma(t)}$ is scheduled, ${\mathit ALG\/}$ can compute $L^t$ and hence
$O^t$ because $L^t$ is $1/m$ times the total processing time of the jobs that have arrived so far.

We next characterize load imbalance. Let $$k = 2i-m \approx (4c-7)m \approx 0.3912m$$ and
$$\alpha = {2(c-1) \over 2c-3} \approx 2.7376.$$ The {\em schedule at time $t$\/} is the one immediately
before $J_{\sigma(t)}$ has to be assigned. The schedule is {\em flat\/} if $l_k^{t-1} < \alpha L_{i+1}^{t-1}$. Otherwise
it is {\em steep\/}. Job $J_{\sigma(t)}$ is scheduled {\em flatly\/} ({\em steeply\/}) if the schedule at
time~$t$ is flat (steep).

${\mathit ALG\/}$ handles each incoming job $J_{\sigma(t)}$, with processing time $p_{\sigma(t)}$, as follows. 
If the schedule at time $t$ is steep, the job is placed on the least loaded machine $M_m^{t-1}$. On the other hand, 
if the schedule is flat, the machines $M_i^{t-1}$, $M_{m-h}^{t-1}$ and $M_m^{t-1}$ are probed in this order. If 
$l_i^{t-1} + p_{\sigma(t)} \leq c\cdot O^t$, then the new machine load on $M_i^{t-1}$ will not violate the desired 
competitiveness. The job is placed on this machine $M_i^{t-1}$. Otherwise, if the latter inequality is violated, 
${\mathit ALG\/}$ checks if a placement on $M_{m-h}^{t-1}$ is safe, i.e.\ if $l_{m-h}^{t-1} + p_{\sigma(t)} \leq c\cdot O^t$. 
If this is the case, the job is put on $M_{m-h}^{t-1}$. Otherwise, $J_{\sigma(t)}$ is finally scheduled on the least loaded
machine~$M_m^{t-1}$. A pseudo-code description of $\mathit{ALG}$ is given below. The job assignment rules are also illustrated
in Figures~\ref{fig:steep} and~\ref{fig:flat}.

\begin{algorithm}[ht]
\caption{The scheduling algorithm ${\mathit ALG}$}\label{alg:1}
\begin{algorithmic}[1]
\State Let $J_{\sigma(t)}$ be the next job to be scheduled.
\If {the schedule at time $t$ is steep}
\State Assign $J_{\sigma(t)}$ to the least loaded machine $M_m^{t-1}$;
 \Else \ // the schedule is flat \If {$l_i^{t-1} + p_{\sigma(t)} \leq c\cdot O^t$} Assign $J_{\sigma(t)}$ to $M_i^{t-1}$;
 \ElsIf {$l_{m-h}^{t-1} + p_{\sigma(t)} \leq c\cdot O^t$} Assign $J_{\sigma(t)}$ to $M_{m-h}^{t-1}$;
 \Else \ Assign $J_{\sigma(t)}$ to $M_m^{t-1}$;
\EndIf \EndIf
\end{algorithmic}
\end{algorithm}

\begin{figure}[H]{
\vspace{-22pt}
\begin{tikzpicture}[scale=0.7]

\def\y{250};
\def\offset{4};
\def\colour{Blue};
\def\colourOrange{Orange}
\foreach \x in {0,...,59} {

\ifnum\x=23
\xdef\colour{Cerulean};
\xdef\offset{16};
\pgfmathparse{\y-14};
  \xdef\y{\pgfmathresult};
\fi
\ifnum\x=41
\xdef\colour{SkyBlue};
\xdef\offset{4};
\pgfmathparse{\y-11};
  \xdef\y{\pgfmathresult};
\fi
\ifnum\x=59
\xdef\colour{Orange};
\fi

\ifx \colour \colourOrange
\draw[preaction={fill, \colour}, pattern=north west lines] (\x/3,0) rectangle (\x/3+1/3,\y/100);
\else
\draw[fill= \colour] (\x/3,0) rectangle (\x/3+1/3,\y/100);
\fi
\pgfmathparse{\y-rand^2*\offset};
  \xdef\y{\pgfmathresult};
}

\draw  (23/3,0) -- (23/3,2.5)node[above] {k};
\draw  (41/3,0) -- (41/3,2.5)node[above] {i};
\end{tikzpicture}
\vspace{-12pt}

\caption{A steep schedule. ${\mathit ALG\/}$ only considers the least loaded machine.} \label{fig:steep}

\begin{tikzpicture}[scale=0.7]

\def\y{250};
\def\colour{Blue};
\def\colourOrange{Orange}
\foreach \x in {0,...,59} {

\ifnum\x=23
\xdef\colour{Cerulean};
\pgfmathparse{\y-5};
  \xdef\y{\pgfmathresult};
\fi
\ifnum\x=40
\xdef\colour{Orange};
\fi
\ifnum\x=41
\xdef\colour{SkyBlue};
\pgfmathparse{\y-5};
  \xdef\y{\pgfmathresult};
\fi
\ifnum\x=57
\xdef\colour{Orange};
\fi
\ifnum\x=58
\xdef\colour{SkyBlue};
\fi
\ifnum\x=59
\xdef\colour{Orange};
\fi

\ifx \colour \colourOrange
\draw[preaction={fill, \colour}, pattern=north west lines] (\x/3,0) rectangle (\x/3+1/3,\y/100);
\else
\draw[fill= \colour] (\x/3,0) rectangle (\x/3+1/3,\y/100);
\fi

\pgfmathparse{\y-rand^2*4};
  \xdef\y{\pgfmathresult};
}

\draw  (23/3,0) -- (23/3,2.5)node[above] {k};
\draw  (41/3,0) -- (41/3,2.5)node[above] {i};
\end{tikzpicture}
}
\vspace{-12pt}
\caption{A flat schedule. The three machines considered by ${\mathit ALG\/}$ are marked for $h=2$.}
\label{fig:flat}
\end{figure}

In the next section we will prove the following theorem, \Cref{te.main}, which uses the notion from \Cref{sec.strongcc}. Lemma~\ref{lem:comp} then immediately gives the main result, \Cref{co.main}.
\begin{theorem}\label{te.main}
${\mathit ALG}$ is nearly $c$-competitive, with $c<1.8478$ defined as above.
\end{theorem}
\begin{corollary}\label{co.main}
${\mathit ALG}$ is $c$-competitive in the random-order model as $m\rightarrow \infty$.
\end{corollary}

\section{Analysis of the algorithm}
\subsection{Analysis basics}
We present some results for the adversarial model so that we can focus on the true random-order analysis of $\mathit{ALG}$ in the next sections. First, recall the three common lower bounds used for online makespan minimization.

\begin{proposition}\label{prop:opt}
For any ${\cal J}$, there holds $OPT({\cal J}) \geq \max\{L,p_{\max}, 2P_{m+1}\}$. Moreover, for any permutation
$J^\sigma $, there holds $O^1\leq O^2\leq \ldots \leq O^n \leq OPT({\cal J})$.
\end{proposition}
\begin{proof}
The optimum makespan $OPT({\cal J})$ cannot be smaller than the average machine load $L$ for the input, even if all
the jobs are distributed evenly among the $m$ machines. Moreover, the job with the largest processing time $p_{\max}$ must be
scheduled non-preemptively on one of the machines in an optimal schedule. Thus $OPT({\cal J}) \geq p_{\max}$. Finally,
among the $m+1$ largest jobs of the input, two must be placed on the same machine in an optimal solution.
Hence $OPT({\cal J}) \geq 2P_{m+1}$. For any permutation $J^\sigma$, the value $O^t$ cannot decrease as jobs $J_t$ arrive.
\end{proof}

For any job sequence ${\cal J} = J_1, \ldots, J_n$, let $R ({\cal J}) = \min \{{L\over p_{\max}}, {p_{\max}\over L}\}$. Intuitively, this measures the complexity of $\cal J$. 
\begin{proposition}\label{prop1}
For any ${\cal J} = J_1,\ldots, J_n$, there holds $\mathit{ALG}({\cal J}) \leq \max\{1+R({\cal J}),c\} OPT({\cal J})$.
\end{proposition}
\begin{proof}
Let ${\cal J} = J_1,\ldots, J_n$ be an arbitrary job sequence and let $J_t$ be the job that defines $\mathit{ALG}$'s makespan.
If the makespan exceeds $c\cdot OPT({\cal J})$, then it exceeds $c\cdot O^t$. Thus $\mathit{ALG}$ placed $J_t$ on machine $M_m^{t-1}$,
cf.\ lines~4 and 5 of the algorithm.  This machine was a least loaded one, having a load of at most $L$. Hence $\mathit{ALG}({\cal J}) \leq
L+p_t\leq L+p_{\max} \leq {L+p_{\max} \over \max\{L,p_{\max}\}} \cdot OPT({\cal J}) = (1 + R({\cal J})) \cdot OPT({\cal J})$.
\end{proof}
Since $R({\cal J}) \leq 1$ we immediately obtain the following result, which ensures that $\mathit{ALG}$ satisfies the first condition
of a nearly $c$-competitive algorithm, see Definition~\ref{def:comp}.
\begin{corollary}\label{lem:adv}
$\mathit{ALG}$ is $2$-competitive in the adversarial model.
\end{corollary}
We next identify a class of {\em plain\/} job sequences that we do not need to consider in the random-order analysis because
$\mathit{ALG}$'s makespan is upper bounded by $c$ times the optimum on these inputs.
\begin{definition}
A job sequence ${\cal J} = J_1,\ldots, J_n$ is called {\em plain\/} if $n\leq m$ or if $R({\cal J}) \leq c-1$. Otherwise it is called
{\em proper}.
\end{definition}
Let ${\cal J} = J_1,\ldots, J_n$ be any job sequence that is processed/scheduled in this order. Observe that if it contains at most $m$ jobs, i.e.\ $n\leq m$, 
and $\mathit{ALG}$ cannot place a job $J_t$ on machines $M_i^{t-1}$ or $M_{m-h}^{t-1}$ because the resulting load would exceed $c\cdot O^t$, 
then the job is placed on an empty machine. Using Proposition~\ref{prop1} we derive the following fact.
\begin{lemma}\label{lem:plain}
For any plain job sequence ${\cal J} = J_1,\ldots, J_n$, there holds $\mathit{ALG}({\cal J}) \leq c \cdot OPT({\cal J})$.
\end{lemma}
If a job sequence ${\cal J}$ is plain (proper), then every permutation of it is. Hence, given Lemma~\ref{lem:plain}, we may concentrate
on proper job sequences in the remainder of the analysis. We finally state a fact that relates to the second condition of a nearly
$c$-competitive algorithm, see again Definition~\ref{def:comp}.
\begin{lemma}\label{lem:loadbasic}
Let ${\cal J} = J_1,\ldots, J_n$ be any job sequence that is scheduled in this order and let $J_t$ be a job that causes $\mathit{ALG}$'s makespan 
to exceed $(c+\varepsilon)OPT({\cal J})$, for some $\epsilon \geq 0$. Then both the load of $\mathit{ALG}$'s least loaded machine at
the time of the assignment as well as $p_t$ exceed $(c-1+\varepsilon)OPT({\cal J})$.
\end{lemma}

\begin{proof}
$\mathit{ALG}$ places $J_t$ on machine $M_m^{t-1}$, which is a least loaded machine when the assignment is done. If $l_m^{t-1}$ or
$p_t$ were upper bounded by $(c-1+\varepsilon)OPT({\cal J})$, then the resulting load would be $l_m^{t-1} + p_t \leq
(c-1+\varepsilon)OPT({\cal J}) + \max\{L,p_t\} \leq (c-1+\varepsilon)OPT({\cal J}) + OPT({\cal J}) = (c+\varepsilon) OPT({\cal J})$.
\end{proof}

\subsection{Stable job sequences}
We define the class of stable job sequences. These sequences are robust in that they will admit an adversarial analysis of $\mathit{ALG}$.
Intuitively, the sequences reveal information on the largest jobs when a significant fraction of the total processing volume $\sum_{t=1}^n p_t$
has been scheduled.  More precisely, one gets an estimate on the processing time of the $h$-th largest job in the entire sequence and
encounters a relevant number of the $m+1$ largest jobs. If a job sequence is unstable, large jobs occur towards the very end of the
sequence and can cause a high makespan relative to the optimum one.

We will show that $\mathit{ALG}$ is adversarially $(c+\varepsilon)$-competitive on stable sequences, for any given $\varepsilon>0$. Therefore,
the definition of stable sequences is formulated for a fixed $\varepsilon>0$. Given ${\cal J}$, let ${\cal J}^\sigma = J_{\sigma(1)},\ldots, J_{\sigma(n)}$
be any permutation of the jobs. Furthermore, for every $j\leq n$ and in particular
$j\in \{h,m+1\}$, the {\em set of the $j$ largest jobs\/} is a fixed set of cardinality $j$ such that no job outside this set has a strictly
larger processing time than any job inside the set.
\begin{definition}\label{def:stable}
A job sequence ${\cal J}^\sigma = J_{\sigma(1)},\ldots, J_{\sigma(n)}$ is {\em stable\/} if the following conditions hold.\\[-20pt]
\begin{itemize}
\item There holds $n>m$.\\[-20pt]
\item Once $L^t \geq (c-1){i\over m}L$, there holds $p_{\max}^t \geq P_h$.\\[-20pt]
\item For every $j\geq i$, the sequence ending once we have $L^t \geq ({j\over m} + {\varepsilon \over 2})L$ contains at least $j+h+2$ many of the
$m+1$ largest jobs in ${\cal J}$.\\[-20pt]
\item The sequence ending right before either (a)~$L^t \geq {i\over m}(c-1)\varepsilon L$ holds or (b)~the $h$-th largest job of ${\cal J}$ is scheduled contains
at least $h+1$ many of the $m+1$ largest jobs in ${\cal J}$.\\[-20pt]
\end{itemize}
Otherwise the job sequence is \emph{unstable}.
\end{definition}
Given $\varepsilon >0$ and $m$, let $P_{\varepsilon}(m)$ be the infimum, over all proper job sequences ${\cal J}$, that a
random permutation of ${\cal J}$ is stable, i.e.\
$$P_{\varepsilon}(m) = \inf_{{\cal J} \ {\rm proper}} \bP_{\sigma\sim S_n}[{\cal J}^\sigma \ {\rm is\ stable}].$$
As the main result of this section we will prove that this probability tends to~1 as $m\rightarrow \infty$.
\begin{mainlemma}\label{ml:1}
For every $\varepsilon >0$, there holds $\displaystyle \lim_{m \rightarrow \infty} P_{\varepsilon}(m) =1$.
\end{mainlemma}
The above lemma implies that for any $\varepsilon >0$ there exists an $m(\varepsilon)$ such that, for all $m\geq m(\varepsilon)$ and
all ${\cal J}$, there holds $\bP_{\sigma\sim S_n}[{\cal J}^\sigma \ {\rm is\ stable}] \geq 1-\varepsilon$. In Section~\ref{sec:adv}
we will show that $\mathit{ALG}$ is $(c+\varepsilon)$-competitive on stable job sequences. This implies
$\bP_{\sigma\sim S_n}[\mathit{ALG}({\cal J}^{\sigma})\geq (c+\varepsilon)OPT({\cal J})] \leq \varepsilon$. Given
Lemma~\ref{lem:adv}, we obtain the following corollary.
\begin{corollary}\label{cor:adv}
If $\mathit{ALG}$ is adversarially $(c+\varepsilon)$-competitive on stable sequences, for every $\varepsilon >0$ and $m\geq m(\varepsilon)$
sufficiently large, then it is nearly $c$-competitive.
\end{corollary}
In the remainder of this section we describe how to establish Main Lemma~\ref{ml:1}. We need some notation. In Section~\ref{sec:alg} the value $L_j^t$ was defined with respect to a fixed job sequence
that was clear from the context. We adopt the notation $L_j^t[{\cal J}^\sigma]$ to make this dependence visible. We adopt a similar
notation for the variables $L$, $P_j^t$, $P_j$, $p_{\max}^t$ and $p_{\max}$. For an input ${\cal J}$ and $\sigma\in S_n$, we will use the notation
$L_j^t[\sigma] = L_j^t[{\cal J}^\sigma]$. Again, we use a similar notation for the variables $P_j^t$ and $p_{\max}^t$.

At the heart of the proof of Main Lemma~\ref{ml:1} is the Load Lemma. Observe that after $t$ time steps in a random permutation of an input
${\cal J}$, each job has arrived with probability $t/n$. Thus the expected total processing time of the jobs seen so far is
$t/n \cdot \sum_{s=1}^n p_s$. Equivalently, in expectation $L^t$ equals $t/n\cdot L$. The Load Lemma proves that this relation holds with
high probability. We set $t = \varphi n$.
\begin{loadlemma}
Given any $\varepsilon>0$ and $\varphi\in (0,1]$, there exists an $m(\varepsilon,\varphi)$ such that for all $m\geq m(\varepsilon,\varphi)$
and all proper sequences ${\cal J}$, there holds
$$
\bP_{\sigma\sim S_n} \left[\left|{L^{\lfloor \varphi n\rfloor}[{\cal J}^{\sigma}]\over \varphi L[{\cal J}^\sigma]} -1\right| \geq \varepsilon \right] \leq \varepsilon.$$
\end{loadlemma}
\begin{proof}
Let us fix a proper job sequence $\CJ$. We use the shorthand $\hat L[\sigma]= \hat L[\CJ^\sigma] = L^{\lfloor \varphi n\rfloor}[\CJ^\sigma]$ and $L=L[\CJ]$.

Let $\delta=\frac{\varphi \varepsilon}{2}$. We will first treat the case that we have $p_\mathrm{max}[\CJ]=1$ and every job size in $\CJ$ is of the form $(1+\delta)^{-j}$, for some $j\ge 0$. Note that we have in particular $c-1\le L\le \frac{1}{c-1}$ because we are working with a proper sequence. For $j\ge 0$ let $h_j$ denote the number of jobs $J_t$ of size $(1+\delta)^{-j}$ and, given $\sigma\in S_n$, let $h_j^\sigma$ denote the number of such jobs $J_t$ that additionally satisfy $\sigma(t)\leq \lfloor \varphi n\rfloor$, i.e.\ they are are among the $\lfloor \varphi n\rfloor$ first jobs in the sequence $\CJ^\sigma$. We now have
$$L= \frac{1}{m}\sum\limits_{j=0}^\infty (1+\delta)^{-j}h_j \ \ \  \mbox{and}
\ \ \ \hat L[\sigma]= \frac{1}{m}\sum \limits_{j=0}^\infty (1+\delta)^{-j}h_j^\sigma.$$

The random variables $h^\sigma_j$ are hypergeometrically distributed, i.e.\ we sample $\lfloor \varphi n\rfloor$ jobs from the set of all $n$ jobs and count the number of times we get one of the $h_j$ many jobs of processing time $(1+\delta)^{-j}$.
Hence, the random variable $h^\sigma_j$ has mean \[\bE[h^\sigma_j]=\frac{\lfloor \varphi n \rfloor}{n} h_j \le \varphi h_j\] and variance
\[ \textrm{Var}[h^\sigma_j] =\frac{h_j\left(n-h_j\right)\lfloor\varphi n \rfloor \left(n-\lfloor\varphi n \rfloor\right) }{n^2(n-1)}
\le h_j  \le (1+\delta)^j mL \le (1+\delta)^j\frac{m}{c-1}.\]
By Chebyshev's inequality we have
\[\bP\bigg[\left|h^\sigma_j-\varphi h_j\right|\ge (1+\delta)^{3j/4}m^{3/4}\bigg] \le (1+\delta)^{-3j/2}\textrm{Var}[h^\sigma_j]m^{-3/2}
\le (1+\delta)^{-j/2}\frac{m^{-1/2}}{c-1}.\]
In particular, by the Union Bound, with probability
\[P(m) =1-\sum\limits_{j=0}^\infty (1+\delta)^{-j/2}\frac{m^{-1/2}}{c-1}
=1-\frac{m^{-1/2}}{\left(1-\sqrt{1+\delta}\right)(c-1)}
=1-O\left(m^{-1/2}\right)\]
we have for all $j$,
\begin{align*}\left|h^\sigma_j-\varphi h_j\right|< (1+\delta)^{3j/4}m^{3/4}.\end{align*}
We conclude that the following holds:
\begin{align*}
\left|\hat L[\sigma]-\varphi L\right|&= \left|\frac{1}{m} \sum\limits_{j=0}^\infty (1+\delta)^{-j}h_j^\sigma -\frac{\varphi}{m}\sum\limits_{j=0}^\infty (1+\delta)^{-j}h_j \right|\\
&\le\sum\limits_{j=0}^\infty (1+\delta)^{-j}\frac{\left|h_j^\sigma -h_j\cdot \varphi\right|}{m}\\
&<\sum\limits_{j=0}^\infty (1+\delta)^{-j/4}m^{-1/4}\\
&=\frac{m^{-1/4}}{\left(1-(1+\delta)^{-1/4}\right)}.\end{align*}
In particular, with probability $P(m)$, we have
\[\left|\frac{\hat L[\sigma]}{\varphi L}-1\right| =\frac{\left|\hat L[\sigma]-\varphi L \right|}{\varphi L}
\le \frac{m^{-1/4}}{\varphi (c-1)\left(1-(1+\delta)^{-1/4}\right)}
=O\left(m^{-1/4}\right).\]
Hence, if we choose $m$ large enough we can ensure that
\begin{equation}\label{eq.Loadlemma}\bP\left[\left|\frac{\hat L[\sigma]}{\varphi L}-1\right|>\frac{\varepsilon}{2}\right] \le 1-P(m)\le \varepsilon.\end{equation}

So far we have assumed that $p_\mathrm{max}[\CJ]=1$ and every job in $\CJ$ has a processing time of $(1+\delta)^{-j}$, for some $j\ge 0$. Now we drop these assumptions.
Given an arbitrary sequence $\CJ$ with $0<p_\mathrm{max}[\CJ]\neq 1$, let $\lfloor\CJ\rfloor$ denote the sequence obtained from $\CJ$ by first dividing every job processing time
by $p_\mathrm{max}[\CJ]$ and rounding every job size down to the next power of $(1+\delta)^{-1}$. We have proven that inequality~\eqref{eq.Loadlemma} holds for
$\lfloor\CJ\rfloor$. The values $L$ and $\hat L[\sigma]$ only change by a factor lying in the interval $[p_\mathrm{max},(1+\delta)p_\mathrm{max})$ when
passing over from $\lfloor\CJ\rfloor$ to $\CJ$.
This implies that
\begin{align*}\left|\frac{\hat L[\CJ^\sigma]}{\varphi L[\CJ]}-\frac{\hat L[\lfloor \CJ\rfloor^\sigma]}{\varphi L[\lfloor \CJ\rfloor]} \right|
&\le \delta \frac{\hat L[\lfloor \CJ\rfloor^\sigma]}{\varphi L[\CJ]}.
\end{align*}
Since $\hat L[\lfloor \CJ\rfloor ^\sigma]\le L[\CJ]$  we obtain 
\begin{align*}\left|\frac{\hat L[\CJ^\sigma]}{\varphi L[\CJ]}-\frac{\hat L[\lfloor \CJ\rfloor^\sigma]}{\varphi L[\lfloor \CJ\rfloor]} \right|&\le \frac{\delta}{\varphi} = \frac{\varepsilon}{2}.\end{align*}
Combining this with inequality~\eqref{eq.Loadlemma} for $\lfloor\CJ\rfloor$ (and the triangle inequality), we obtain
\[\bP\left[\left|\frac{\hat L[\CJ^\sigma]}{\varphi L[\CJ]}-1\right|>\varepsilon\right]
\le \bP\left[\left|\frac{\hat L[\lfloor\CJ\rfloor^\sigma]}{\varphi L[\lfloor\CJ\rfloor]}-1\right|>\frac{\varepsilon}{2}\right] \le \varepsilon.\]
Thus the lemma follows.
\end{proof}
We note that the Load Lemma does not hold for general sequences. A counterexample is a job sequence in which one job carries all the load, while
all the other jobs have a negligible processing time. The proof of the Load Lemma relies on a lower bound of $R({\cal J})$, which is $c-1$
for proper sequences.

We present two consequences of the Load Lemma that will allow us to prove that stable sequences reveal information on the largest jobs when a
certain processing volume has been scheduled. Consider a proper ${\cal J}$. Given ${\cal J}^\sigma = J_{\sigma(1)},\ldots, J_{\sigma(n)}$ and 
$\varphi>0$, let $N(\varphi)[{\cal J}^\sigma]$ be the number of jobs $J_{\sigma(t)}$ that are among the $m+1$ largest jobs in ${\cal J}$ and
such that $L^t \leq \varphi L$. 
\begin{restatable}{lemma}{llI}\label{le.ll1}
Let $\varepsilon>0$ and $\varphi\in (0,1]$. Then there holds
$$\lim_{m\rightarrow \infty} \inf_{{\cal J} \ {\rm proper}} \bP_{\sigma\sim S_n} \left[N(\varphi+\varepsilon)[{\cal J}^{\sigma}] \geq
\lfloor \varphi m\rfloor +h+2\right]=1.$$
\end{restatable}
We will just state the core argument here and leave the rather technical proof to \Cref{app}.
\begin{proof}[Proof sketch] The Load Lemma basically matches load ratios $L^t/L$ with ratios $t/n$ on the time line of job arrivals, up to some
margin of error. We can then infer that at least $\lfloor \varphi m\rfloor +h+1$ of the $m+1$ largest jobs are among the first $(\varphi+\varepsilon)n$ 
jobs in a job sequence ${\cal J}^\sigma$, with a probability that tends to~1 as $m\rightarrow \infty$. In expectation 
$(\varphi+\varepsilon)(m+1)$ of the $m+1$ largest jobs occur in this prefix, which is strictly more than $\lfloor \varphi m\rfloor +h+1$, for $m$ large enough. 
Formally, we show that (a slight variant of) the random variable $N(\varphi +\varepsilon)[{\cal J}^\sigma]$ 
is hypergeometrically distributed and has variance at most $m+1$. Using Chebyshev's inequality we derive Lemma~\ref{le.ll1}.
\end{proof}
\begin{lemma}\label{le.ll2}
Let $\varepsilon>0$ and $\varphi\in (0,1]$. Then there holds
$$\lim_{m\rightarrow \infty} \inf_{{\cal J} \ {\rm proper}} \bP_{\sigma\sim S_n} \left[\forall_{\tilde{\varphi}\geq \varphi} \
N(\tilde{\varphi}+\varepsilon)[{\cal J}^{\sigma}] \geq
\lfloor \tilde{\varphi} m\rfloor +h+2\right]=1.$$
\end{lemma}
\begin{proof}
Let us fix any proper sequence $\CJ$ and set \[\Lambda=\left\{1-\frac{\varepsilon}{2}j\mid j\in\BN, \varphi\le 1-\frac{\varepsilon}{2}j \right\}\] which is a finite set whose size only depends on $\varepsilon$ and $\varphi$.
Given $\tilde\varphi\ge\varphi$, let $u(\tilde\varphi)$ be the smallest element in $\Lambda$ greater or equal to $\tilde\varphi$. Then\[\tilde\varphi\le u( \tilde\varphi) \le\tilde\varphi+\frac{\varepsilon}{2}\] and if we have
\[N(\tilde\varphi+\varepsilon)[{\cal J}^\sigma] < \lfloor \tilde\varphi m \rfloor +h+2\]
there holds
\[N\left(u(\tilde\varphi)+\frac{\varepsilon}{2}\right)[{\cal J}^\sigma] < \lfloor \tilde\varphi m \rfloor+h+2\le \lfloor u(\tilde\varphi) m \rfloor+h+2.\]
In particular, in order to prove the lemma it suffices to verify that
\[\lim\limits_{m\rightarrow\infty}\inf_{\CJ\textrm{ proper}}\bP_{\sigma\sim S_n}\left[\forall_{\tilde \varphi \in\Lambda} N\left(\tilde\varphi+\frac{\varepsilon}{2}\right)[\CJ^\sigma] \ge \lfloor \tilde\varphi m \rfloor+h+2\right]=1.\]
The latter is a consequence of applying \Cref{le.ll1} to all $\tilde\varphi\in\Lambda$ and the Union Bound.
\end{proof}

We can now conclude the main lemma of this section:

\begin{proof}[Proof of Main Lemma~\ref{ml:1}:]
A proper job sequence is stable if the following four properties hold.\\[-15pt]
\begin{itemize}
\item Once $L^t \ge (c-1)\frac{i}{m}\cdot L$ we have $p^t_\mathrm{max} \ge P_h$.\\[-20pt]
\item For every $j\ge i$ the sequence ending once we have $L^t \ge\left (\frac{j}{m}+\frac{\varepsilon}{2}\right) L$ contains at least $j+h+2$ of the $m+1$ largest jobs.\\[-20pt]
\item The sequence ending right before $L^t \ge \frac{i}{m}(c-1)\varepsilon L$ holds contains at least $h+1$ of the $m+1$ largest jobs.\\[-20pt]
\item The sequence ending right before the first of the $h$ largest jobs contains at least $h+1$ of the $m+1$ largest jobs.
\end{itemize}
By the Union Bound we may consider each property separately and prove that it holds with a probability that tends to $1$ as $m\rightarrow \infty$.

Let $\varphi=(c-1)\frac{i}{m}$ and choose $\varepsilon >0$. By the Load Lemma, for $m\geq m(\varepsilon,\varphi)$, after $t=\left\lfloor\varphi n\right\rfloor$ jobs
of a proper job sequence ${\cal J}^\sigma$ have been scheduled, there holds $L^t\le(c-1)\frac{i}{m}\cdot L$ with probability at least $1-\varepsilon$.
Observe that $\varphi$ is a fixed problem parameter so that $m(\varepsilon,\varphi)$ is determined by $\varepsilon$.
The probability of any particular job being among the first $t$ jobs in $\CJ^\sigma$ is
$\lfloor \varphi n\rfloor/n$. Thus $p^t_\mathrm{max} \ge P_h$ holds with probability at least
$1-(1- \lfloor \varphi n\rfloor/n)^h$. Since ${\cal J}^\sigma$ is proper, we have $n>m$. Furthermore, $h=h(m) \in \omega(1)$.
Therefore, the probability that the first property holds tends to~1 as $m\rightarrow \infty$.

The second property is a consequence of \Cref{le.ll2} with $\varphi=\frac{i}{m}$. The third property follows from \Cref{le.ll1}.
We need to choose the $\varepsilon$ in the statement of the lemma to be $\frac{i}{m}(c-1)\varepsilon$.
Finally we examine the last property. In ${\cal J}^\sigma$ we focus on the positions of the $m+1$ largest jobs. Consider any of the $h$ largest jobs.
The probability that it is preceded by less than $h+1$ of the $m+1$ largest jobs is $(h+1)/(m+1)$. Thus the probability of the fourth property not to hold
is at most $h(h+1)/(m+1)$. Since $h\in o(\sqrt{m})$, the latter expression tends to~0 as $m\rightarrow \infty$.
\end{proof}

\subsection{An adversarial analysis}\label{sec:adv}
\vspace{-5pt}
In this section we prove the following main result.
\begin{mainlemma}\label{ml:2}
For every $\varepsilon >0$ and $m\geq m(\varepsilon)$ sufficiently large, $\mathit{ALG}$ is adversarially $(c+\varepsilon)$-competitive on
stable job sequences.
\end{mainlemma}
Consider a fixed $\varepsilon>0$. Given Lemma~\ref{lem:adv}, we may assume that $0<\varepsilon < 2-c$. Suppose that there was a stable job sequence $\CJ^\sigma$ 
such that $\mathit{ALG}(\CJ^\sigma) > (c+\varepsilon)OPT(\CJ^\sigma)$. We will derive a contradiction, given that $m$ is large. In order to simplify 
notation, in the following let $\CJ = \CJ^\sigma$ be the stable job sequence violating the performance ratio of 
$c+\varepsilon$. Let $\CJ = J_1,\ldots, J_n$ and $OPT = OPT(\CJ)$.

Let $J_{n'}$ be the first job that causes $\mathit{ALG}$ to have a makespan greater than $(c+\varepsilon)OPT$ and let $b_0=l_m^{n'-1}$ be the load of the least loaded machine $M_{m}^{n'-1}$ right before $J_{n'}$ is scheduled on it. The makespan after $J_{n'}$ is scheduled, called the \textit{critical makespan},
is at most $b_0+p_{n'}\le b_0 +OPT$. In particular $b_0>(c-1+\varepsilon) OPT$ as well as $p_{n'}>(c-1+\varepsilon) OPT$, see Lemma~\ref{lem:loadbasic}.
Let
$$\textstyle \lambda_\mathrm{start}=\frac{c-1}{1+2c(2-c)}\approx 0.5426 \ \ \ {\rm and} \ \ \ \lambda_\mathrm{end}= \frac{1}{2(c-1+\varepsilon)}\approx 0.5898.$$
There holds $\lambda_\mathrm{start}< \lambda_\mathrm{end}$. The critical makespan of $\mathit{ALG}$ is bounded by
$b_0+ OPT< (1+\frac{1}{c-1+\varepsilon})b_0 = (c+\varepsilon)\frac{b_0}{c-1+\varepsilon}=(c+\varepsilon)2\lambda_\mathrm{end}b_0.$
Since $\mathit{ALG}$ does not achieve a performance ratio of $c+\varepsilon$ on $\CJ$ we have
\begin{equation}\label{eq:contr}
P_{m+1}\le OPT /2 <\lambda_\mathrm{end}b_0.
\end{equation}
Our main goal is to derive a contradiction to this inequality.

\paragraph*{The impact of the variable $\boldsymbol{P_h}$:}A new, crucial aspect in the analysis of $\mathit{ALG}$ is $P_h$, the processing time of the $h$-th
largest job in the sequence $\CJ$. Initially, when the processing of $\CJ$ starts, we have no information on $P_h$ and can only infer
$P_{m+1}\geq \lambda_\mathrm{start}b_0$. The second property in the definition of stable job sequences ensures that $p^t_\mathrm{max}\geq P_h$ once the
load ratio $L^t/L$ is sufficiently large. Note that $\mathit{ALG}$ then also works with this estimate because $P_h \leq p^t_\mathrm{max} \leq O^t$.
This will allow us to evaluate the processing time of flatly scheduled jobs. In order prove that $P_{m+1}$ is large, we will relate $P_{m+1}$ and
$P_h$, i.e.\ we will lower bound $P_{m+1}$ in terms of $P_h$ and vice versa. Using the relation we can then conclude $P_{m+1} \geq \lambda_\mathrm{end}b_0$.
In the analysis we repeatedly use the properties of stable job sequences and will explicitly point to it when this is the case.

We next make the relationship between $P_h$ and $P_{m+1}$ precise. Given $0<\lambda$, let $f(\lambda)=2c\lambda-1$ and given $w>0$, let $g(w)=(c (2c-3)-1)w+4-2c
\approx 0.2854\cdot w+0.3044$. We set $g_b(\lambda)=g\left(\frac{\lambda}{b}\right) b$ and $f_b(w)=f\left(\frac{w}{b}\right) b$, for any $b>0$. Then we will
lower bound $P_{m+1}$ by $g_{b_0}(P_h)$ and $P_h$ by $f_{b_0}(P_{m+1})$. We state two technical propositions.

\begin{proposition}\label{pro.1.1}
For $\lambda>\lambda_\mathrm{{start}}$, we have $g(f(\lambda))>\lambda$.
\end{proposition}
\begin{proof}
Consider the function
\begin{align*}F(\lambda)=g(f(\lambda))-\lambda&=(c (2c-3)-1)(2c\lambda-1)+4-2c-\lambda\\
&=(4c^3-6c^2-2c-1)\lambda -2c^2+c+5\\
&\approx 0.05446\cdot \lambda +0.01900.
\end{align*}
The function $F$ is linear and strictly increasing in $\lambda$. Hence for the proposition to hold it suffices to verify that $F(\lambda_\mathrm{start})\ge 0$. We can now compute that
$F(\lambda_\mathrm{start})
\approx 0.04865 >0.$
\end{proof}

\begin{proposition}\label{pro.1.2}
For $0<\varepsilon \le 1$, we have $g(1-\varepsilon)> \lambda_\mathrm{{end}}$.
\end{proposition}

Note that the following proof determines the choice of our competitive ratio $c$, which was chosen minimal such that $Q[c]=4c^3-14c^2+16c-7\ge 0$.

\begin{proof}
We calculate that
\begin{align*}
g(1-\varepsilon)-\lambda_\mathrm{end} &= (c(2c-3)-1)(1-\varepsilon)+4-2c-\frac{1}{2(c-1+\varepsilon)}\\
&=\frac{2(c-1+\varepsilon)(2c^2-5c+3-(2c^2-3c-1)\varepsilon)-1}{2(c-1+\varepsilon)}\\
&=\frac{4c^3-14c^2+16c-7 +(4-2c)\varepsilon - 2(2c^2-3c-1)\varepsilon^2}{2(c-1+\varepsilon)}.
\end{align*}
Recall that $Q[c]=4c^3-14c^2+16c-7=0$. For $0<\varepsilon \le 1$ we have
\[(4-2c)\varepsilon - (2c^2-3c-1)\varepsilon^2 \approx 0.3044\cdot\varepsilon -0.2854\cdot\varepsilon^2 >0.\]
Thus we see that $g(1-\varepsilon)-\lambda_\mathrm{end}>0$ and can conclude the lemma.
\end{proof}
\subsubsection{Analyzing large jobs towards lower bounding $\boldsymbol{P_h}$ and $\boldsymbol{P_{m+1}}$}
Let $b>(c-1+\varepsilon) OPT$ be a value such that immediately before $J_{n'}$ is scheduled at least $m-h$ machines have a load of at least $b$.
Note that $b=b_0$ satisfies this condition but we will be interested in larger values of $b$ as well. We call a machine \textit{$b$-full} once
its load is at least $b$; we call a job $J$ a \textit{$b$-filling job} if it causes the machine it is scheduled on to become $b$-full.
We number the $b$-filling jobs according to their order of arrival $J^{(1)}, J^{(2)}, \ldots $ and let $t(j)$ denote the time of arrival of the $j$-th filling job~$J^{(j)}$.

Recall that our main goal is to show that $P_{m+1}\ge \lambda_\mathrm{end}b_0$ holds. To this end we will prove that the \textit{$b_0$-filling jobs}
have a processing time of at least $\lambda_\mathrm{end}b_0$. As there are $m$ such jobs, the bound on $P_{m+1}$ follows by observing that $J_{n'}$
arrives after all \textit{$b_0$-filling jobs} are scheduled and that its processing time exceeds $\lambda_\mathrm{end}b_0$ as well. 
In fact, since $OPT\geq b_0$, we have
\begin{equation}\label{eq.lastjob}
p_{n'}> (c-1) OPT > 0.847\cdot OPT > \lambda_\mathrm{end}b_0 \approx 0.5898\cdot b_0.
\end{equation}
We remark that different to previous analyses in the literature we do not solely rely on lower bounding the processing time of filling jobs.
By using the third property of stable job sequences, we can relate load and the size of the $(m+1)$-st largest job at specific points
in the time horizon, cf.\ \Cref{le.techlemma}.

In the following we regard $b$ as fixed and omit it from the terms {\em filling job\/} and {\em full}. Let
$\lambda=\max \{\lambda_\mathrm{start}b,\min\{g_{b}\left(P_h\right),\lambda_\mathrm{end}b\}\}$. We call a job \textit{large} if it has
a processing time of at least $\lambda$. Let $\tilde t=t(m-h)$ be the time when the $(m-h)$-th filling job arrived.
The remainder of this section is devoted to showing the following important Lemma~\ref{le.fill}. Some of the underlying
lemmas, but not all of them, hold if $m\geq m(\varepsilon)$ is sufficiently large. We will make the dependence clear.
\begin{lemma}\label{le.fill}
At least one of the following statements holds:
\begin{itemize}
\itemsep0em
\item All filling jobs are large.
\item If $m\geq m(\varepsilon)$, there holds $P_{m+1}^{\tilde t} \ge \lambda=\max\{\lambda_\mathrm{{start}}b,\min\{g_{b}\left(P_h\right),\lambda_\mathrm{{end}}b\}\}$, i.e.\
there are at least $m+1$ large jobs once the $(m-h)$-th filling job is scheduled.
\end{itemize}
\end{lemma}

Before we prove the lemma we derive two important implications towards a lower bound of $P_{m+1}$.
\begin{corollary}\label{le.fill2}
We have
$P_{m+1} \ge \lambda=\max\{\lambda_\mathrm{{start}}b_0,\min\{g_{b_0}\left(P_h\right),\lambda_\mathrm{{end}}b_0\}\}.$
\end{corollary}

\begin{proof}
Apply the previous lemma, taking into account that $b\geq b_0$, and use that there are $m$ many $b_0$-filling jobs followed by $J_{n'}$. The latter has size at least $\lambda$ by inequality~\eqref{eq.lastjob}.
\end{proof}

We also want to lower bound the processing time of the $(m+1)$-st largest job at time $\tilde t$. However, at that time only $m-h$ filling jobs
have arrived. The next lemma ensures that, if additionally $P_{h}$ is not too large, this is not a problem.
\begin{corollary}\label{le.fill3}
If $P_h \le (1-\varepsilon)b$ and $m\geq m(\varepsilon)$, the second statement in \Cref{le.fill} holds, i.e.\
$P_{m+1}^{\tilde t} \ge \lambda=\max\{\lambda_\mathrm{{start}}b,\min\{g_{b}\left(P_h\right),\lambda_\mathrm{{end}}b\}\}.$
\end{corollary}

The proof of the lemma makes use of the fourth property of stable job sequences. In particular we would not expect such a result to hold in the adversarial model.

\begin{proof}
We will show that the first statement in \Cref{le.fill} implies the second one if $P_h \le (1-\varepsilon)b$ holds.  In order to conclude the second statement it suffices to verify
that at least $m+1$ jobs of processing time $\lambda$ have arrived until time~$\tilde t$. By the first statement we know that there were $m-h$ large filling jobs coming before time $\tilde t$.
Hence it is enough to verify that $h+1$ large jobs arrive (strictly) before the first filling job $J$.

To show that there are $h+1$ jobs with a processing time of at least $P_{m+1}$ before the first filling job $J$, we use the last property of stable job sequences.
If $J$ is among the $h$ largest jobs, we are done immediately by the condition. Else $J$ had size at most $P_h\le (1-\varepsilon)b$. Assume $J=J_t$ was scheduled on the machine $M_j^{t-1}$, for $j\in\{i,m-h,m\}$, and let $l=l_j^{t-1}$ be its load before $J$ was scheduled. Because $J$ is a filling job we have \[l\ge b-P_h \ge \varepsilon b \ge \varepsilon(c-1)OPT.\] In particular, before $J$ was scheduled, the average load at that time was at least \[\frac{jl}{m}\ge \frac{il}{m}\ge\varepsilon \frac{i}{m}(c-1)OPT.\] Again,
by the last property of stable job sequences, at least $h+1$ jobs of processing time at least~$P_{m+1}$ were scheduled before this was the case.
\end{proof}

We introduce late and early filling jobs. We need a certain condition to hold, see \Cref{le.techlemma}, in order to show that the early filling jobs
are large. We show that if this condition is not met, the fact that the given job sequence is stable ensures that $P_m^{\tilde t} \ge \lambda$.

Let $s$ be chosen maximal such that the $s$-th filling job is scheduled steeply. If $s\le i$, then set $s=i+1$ instead. We call all filling jobs $J^{(j)}$
with $j>i$ that are scheduled flatly \textit{late filling jobs}. All other filling jobs are called \textit{early filling jobs}. In particular the
job $J^{(s+1)}$ and the filling jobs afterwards are late filling jobs. The following proposition implies that the fillings jobs after $J^{(m-h)}$,
if they exist, are all late, i.e.\ scheduled flatly.
\begin{proposition}\label{pro.endflat}
We have $s \le m-h$ if $m\geq m(\varepsilon)$.
\end{proposition}
\begin{proof}[Proof of \Cref{pro.endflat}]
Let $\tilde h < h$ and $t=t(m-\tilde h)$ be the time the $(m-\tilde h)$-th filling job $J$ arrived. We need to see that $J$ was scheduled flatly. Assume that was not the case. We know that for $j\le m-\tilde h$ we have $l_j^{t-1}\ge b> (c-1+\varepsilon)OPT$. In particular we have
\[ L_{i+1}^{t-1} =\frac{1}{m-i}\sum\limits_{j=i+1}^m l_{j}^{t-1} > \frac{m-i-h-1}{m-i} (c-1+\varepsilon)OPT
\ge (c-1)OPT.\]
For the last inequality we need to choose $m$ large enough. If the schedule was steep at time $t$, then we had for every $j\le k$ \[l_j^{t-1}
\ge l_k^{t-1} \ge \alpha (c-1)OPT = \frac{2(c-1)^2}{2c-3}OPT.\]
But then the average load at time $t-1$ would be:
\begin{align*}L^{t-1}&=\frac{1}{m}\sum\limits_{j=1}^m l_j^{t-1}
\\& >\frac{k\frac{2(c-1)^2}{2c-3}+(m-h-k)(c-1)}{m}OPT
\\& \ge\frac{((4c-7)m+2h)\frac{2(c-1)^2}{2c-3}+(m-(4c-7)m-h)(c-1)}{m}OPT
\\&\approx 1.3247 \cdot OPT.\end{align*}
For the second inequality we need to observe that we have $k\ge (4c-7)m+2h$ and that the previous term decreases if we decrease $k$. One also can check that the second last term is minimized if $h=0$.

But now we have shown $L^{t-1}>OPT$, which is a contradiction. Hence the schedule could not have been steep at time $t-1$.
\end{proof}

We need a technical lemma. For any time $t$, let $\overline{L}_s^t = {1\over m-h-s+1}\sum_{j=s}^{m-h} l_j^t$ be the average load on 
the machines numbered $s$ to $m-h$. 
\begin{restatable}{lemma}{letechDts}\label{le.tech.Dts}
If $\overline{L}_s^{t(s)-1}\ge \alpha^{-1} b$ holds and $m\geq m(\varepsilon)$, we have $L^{t(s)-1}>\left(\frac{s}{m}+\frac{\varepsilon}{2}\right)\cdot L$.
\end{restatable}
This lemma comes down to a mere computation. While being simple at its core, we have to account for various small error terms. These arise in three ways. Some are inherent to the properties of stable sequences. Others arise from the rounding involved in the definition of certain numbers, $i$ in particular. Finally, the small number $h$ introduces such an error. While all these errors turn out to be negligible, rigorously showing so is technical. We thus leave the proof to \Cref{app}. The reader highly familiar with other works on online makespan minimization may have wondered about our different choice of the value $i$. It should be noted that $i$ is chosen maximal, such that \Cref{le.tech.Dts} holds true.
\begin{lemma}\label{le.techlemma}
If the late filling jobs are large, $\overline{L}_s^{t(s)-1}\ge \alpha^{-1}b$ and $m\geq m(\varepsilon)$, we have $P_{m+1}^{\tilde t}\ge \lambda$.
\end{lemma}
\begin{proof}
Assume that the conditions of the lemma hold. By \Cref{le.tech.Dts} we have $L^{t(s)-1}>\left(\frac{s}{m}+\frac{\varepsilon}{2}\right)\cdot L.$
By the third property of stable sequences, at most $m+1-(s+h+2)=m-s-h-1$ of the largest $m+1$ jobs appear in the sequence starting
after time $t(s)-1$. However, this sequence contains $m-h-s$ late filling jobs. Thus there exists a
late filling job that is not among the $m+1$ largest jobs. As it has a processing time of at least $\lambda$, by the assumption of the lemma, 
$P_{m+1}\ge \lambda$ holds.

Now consider the $m+1$ largest jobs of the entire sequence that arrive before $J^{(s)}$ as well as the jobs $J^{(s+1)},\ldots, J^{(m-h)}$. 
There are at least $s+h+2$ of the former 
and $m-h-s$ of the latter. Thus we have found a set of at 
least $m+1$ jobs arriving before (or at) time $\tilde t=t(m-h)$. Moreover,
we argued that all these jobs have a processing time of at least $\lambda$. Hence $P_{m+1}^{\tilde t}\ge \lambda$ holds true.
\end{proof}

We are ready to evaluate the processing time of filling jobs to prove \Cref{le.fill}, which we will do in the following two lemmas.
\begin{lemma}\label{le.latefill}
The processing time of any late filling jobs strictly exceeds $\max\{\lambda_\mathrm{{start}}b,g_b(P_h)\}$.
\end{lemma}

\begin{proof}
Let $j\ge i+1$ such that $J^{(j)}$ was scheduled flatly. Set $t=t(j)-1$ and $l=l_i^t$. Because at least $i$ machines were full, we have have
$L^t\ge b\cdot {i\over m} \ge (c-1){i\over m} OPT \geq (c-1){i\over m}L$. Hence by
Definition~\ref{def:stable} we have $p^{t}_\mathrm{max}\ge P_h$.

Let $\tilde\lambda=\max\{\lambda_\mathrm{start}b,g_b(P_h)\}$. We need to show that $J^{(j)}$ has a processing time strictly greater than $\tilde\lambda$. If we have $l_{m-h}^t<b-\tilde\lambda$,
then this was the case because $J^{(j)}$ increased the load of some machine from a value smaller than $b-\tilde\lambda$ to $b$. Hence let us assume that we have $l_{m-h}^t\ge b-\tilde\lambda$.
In particular we have
\begin{align*}
L^t \ge \frac{(j-1)l+(m-j-h+1)(b-\tilde\lambda)}{m}.
\end{align*}
By the definition of a late filling job, $J^{(j)}$ was scheduled flatly. In particular, it would have been scheduled on machine $M_i^t$ (which was not the
case) if any of the following two inequalities did not hold:
\begin{itemize}
\item $p_t + l > cp^{t}_\mathrm{max} \ge cP_h$
\item $p_t + l > cL^t$
\end{itemize}
If $l\le cP_h-\tilde\lambda$ held true, we get $p_t>\tilde\lambda$ from the first inequality. Thus we only need to treat the case that $l> cP_h-\tilde\lambda$ held true.
We also know that we have $l\ge b$, because the $i$-th machine is full. Hence we may assume that
\[l\ge \max\{b,cP_h-\tilde\lambda\}.\]
In order to derive the lemma we need to prove that $p_t-\tilde\lambda>0$ holds. Using the second inequality we get
\[ p_t -\tilde\lambda > cL^t -l-\tilde\lambda  \ge c \frac{(j-1)l+(m-j-h+1)(b-\tilde\lambda)}{m}-l -\tilde\lambda.\]
Using that $(2c-3)m+h<i<j-1$ and $b-\tilde\lambda<b\le l$ hold, the previous term does not increase if we replace $j-1$ by $(2c-3)m+h$. The resulting term is
\begin{align*}p_t -\tilde\lambda> c\frac{((2c-3)m+h)l+(m-(2c-3)m-2h+1)(b-\tilde\lambda)}{m}-l-\tilde\lambda.\end{align*}
Now let us observe that we have $l\ge b \geq 2(b-\lambda_\mathrm{start} b) \ge 2(b-\tilde\lambda)$. Hence the previous term is minimized if we set $h=0$. We get
\begin{align*}
p_t -\tilde\lambda> c\big[(2c-3)l+(1-(2c-3))(b-\tilde\lambda)\big]-l-\lambda.\end{align*}
As $c(2c-3)-1\approx 0.2584>0$ the above term does not increase if we replace $l$ by either value: $b$ or $cP_h-\tilde\lambda$.

If we have $\tilde\lambda=\lambda_\mathrm{start}b$, we choose $l=b$ and get
\begin{align*}
p_t -\lambda_\mathrm{start}b&>c\big[(2c-3)b+(1-(2c-3))(b-\lambda_\mathrm{start}b)\big]-b-\lambda_\mathrm{start}b\\
&= (c-1)b-(1+2c(2-c))\lambda_\mathrm{start}b\\
&= (c-1)b-(c-1)b\\
&=0.
\end{align*}
The third equality uses the definition of $\lambda_\mathrm{start}$. The lemma follows if $\tilde\lambda=\lambda_\mathrm{start}b$.

Otherwise, if $\tilde\lambda=g_b(P_h)$, we choose $l=cP_h-\tilde\lambda$ and get
\begin{align*}
p_t -\tilde\lambda&>c\big[(2c-3)(cP_h-\tilde\lambda)+(1-(2c-3))(b-\tilde\lambda)\big]-(cP_h-\tilde\lambda)-\tilde\lambda\\
&=(c^2(2c-3)-c)P_h+(c(4-2c))b-cg_b(P_h)\\
&=0.
\end{align*}
Here the last equality follows from the definition of $g_b$. The lemma follows in the case $\tilde\lambda=g_b(P_h)$.
\end{proof}
\begin{lemma}\label{le.earlyfill}
If $\overline{L}_s^{t(s)-1}< \alpha^{-1}b$ holds, the early filling jobs have a processing time of at least $\lambda_\mathrm{{end}}b$.
\end{lemma}
Before proving \Cref{le.earlyfill} let us observe the following, strengthening its condition.

\begin{lemma}\label{rem.Dvalsincrease}
We have
\[L_{i+1}^{t(i+1)-1}\le L_{i+2}^{t(i+2)-1}\le \ldots  L_s^{t(s)-1}.\]
\end{lemma}

\begin{proof}
Let $i+1 \le j<s$. It suffices to verify that
\[L_j^{t(j)-1}\le L_{j+1}^{t(j)} \le L_{j+1}^{t(j+1)-1}.\]
The second inequality is obvious because for every $r$ the loads $l_r^t$ can only increase as $t$ increases. For the first inequality we note that by definition the job $J^{(j)}$ was scheduled steeply and hence on a least loaded machine. This machine became full. Thus it is not among the $m-j$ least loaded machines at time $t(j)$. In particular $L_{j+1}^{t(j)}$, the average over the $m-j$ smallest loads at time $t(j)$, is also the average of the $m-j+1$ smallest loads excluding the smallest load at time $t(j)-1$. Therefore it cannot be less than $L_j^{t(j)-1}$.
\end{proof}

\begin{proof}[Proof of \Cref{le.earlyfill}]
Let $i<j\le s$ such that $J^{(j)}$ was an early filling job. By \Cref{rem.Dvalsincrease}  we have $L_j^{t(j)-1}\le L_s^{t(s)-1}< \alpha^{-1}b=b-\frac{b}{2(c-1)}< b-\lambda_\mathrm{end}b$. By definition $J^{(j)}$ was scheduled on a least loaded machine $M^{t(j)-1}_m$ which had load less than $L_j^{t(j)-1}< b-\lambda_\mathrm{end}b$ before
and at least $b$ afterwards because it became full. In particular $J^{(j)}$ had size $\lambda_\mathrm{end}b$.

For $k<j\leq i$ the job $J^{(j)}$ is scheduled steeply because we have by \Cref{rem.Dvalsincrease}
\[ l_k^{t(j)-1} \ge b > \alpha L_{s}^{t(s)-1}
\ge\alpha L_{i+1}^{t(i+1)-1}
\ge \alpha L_{i+1}^{t(j)-1}.\]
Thus for $k<j\leq i$ the job $J^{(j)}$ is scheduled on the least loaded machine $M_{m}^{t(j)-1}$, whose load $l_{m}^{t(j)-1}$ is bounded by
\[l_{m}^{t(j)-1} \le L_{i+1}^{t(j)-1} \le L_{s}^{t(s)-1}<\alpha^{-1}b=b-\frac{b}{2(c-1)}< b-\lambda_\mathrm{end}b.\]
Hence the job $J^{(j)}$ had a size of at least $\lambda_\mathrm{end}b$.
We also observe that we have \[l_i^{t(k)-1} \le l_{i+1}^{t(k+1)-1} \le \ldots \le l_{i+(m-i)}^{t(k+(m-i))-1}=l_m^{t(i)-1}<b-\lambda_\mathrm{end}b.\] In particular for $1\le j \le k$ any filling job $J^{(j)}$ filled a machine with a load of at most $\max\{l_m^{t(k)} ,l_i^{t(k)}\}=l_i^{t(k)}<b-\lambda_\mathrm{end}b$.
Hence it had a size of at least $\lambda_\mathrm{end}b$.
\end{proof}

We now conclude the main lemma of this subsection, \Cref{le.fill}.

\begin{proof}[Proof of \Cref{le.fill}]
By \Cref{le.latefill}, all late filling jobs are large. We distinguish two cases depending on whether or not $\overline{L}_s^{t(s)-1}< \alpha^{-1}b$ holds.
If it does, all filling jobs are large by \Cref{le.earlyfill} and the first statement in \Cref{le.fill} holds.
Otherwise, the second statement in \Cref{le.fill} holds by \Cref{le.techlemma}.
\end{proof}

\subsubsection{Lower bounding $\boldsymbol{P_h}$ and $\boldsymbol{P_{m+1}}$}
In this section we establish the following relations on $P_h$ and $P_{m+1}$.
\begin{lemma}\label{le.P_h}
There holds  $P_h>(1-\varepsilon)b_0$ or $P_{m+1}\ge \lambda_\mathrm{{end}}b_0$ if $m\geq m(\varepsilon)$.
\end{lemma}

For the proof we need a way to lower bound the processing time of a job $J_t$ depending on $P_{m+1}^t$:
\begin{lemma}\label{le.2}
Let $J_t$ be any job scheduled flatly on the least loaded machine and let $b=l_{m-h}^{t-1}$ be the load of the $(h+1)$-th least
loaded machine. Then $J_t$ has a processing time of at least $f_b(P_{m+1}^{t})$.
\end{lemma}
\begin{proof}
From the fact that $J_t$ was not scheduled on the $(h+1)$-th least loaded machine $M_{m-h}^t$ we derive that
$p_t>c\cdot O^{t}-b\ge c\cdot P_{m+1}^{t}-b=f_b(P_{m+1}^{t})$
holds.
\end{proof}

\begin{proof}[Proof of \Cref{le.P_h}]
Assume for a contradiction  that we had $P_h \le (1-\varepsilon)b_0$. Let $J=J_t$ be the smallest among the $h$ last $b_0$-filling jobs. Then $J$ has
a processing time $p\le P_h$. We want to derive a contradiction to that.
Let $b_1=l_{m-h}^{t-1}$ be the load of the $(m-h)$-th machine right before $J$ was scheduled. Because this machine was $b_0$-full at that time we know
that $b_1\ge b_0>(c-1+\varepsilon) OPT$ holds and it makes sense to consider $b_1$-filling jobs. Let $\tilde t$ be the time the $(m-h)$-th $b_1$-filling job arrived. By \Cref{le.fill} we have $P_{m+1}^{\tilde t}\ge \lambda=\max\{\lambda_\mathrm{start}b_1,\min\{g_{b_1}\left(P_h\right),\lambda_\mathrm{end}b_1\}\}.$

If we have $\lambda=\lambda_\mathrm{end}b_1\ge \lambda_\mathrm{end}b_0$ we have already proven $P_{m+1}\ge \lambda_\mathrm{end}b_0$ and the lemma follows. So we are left to treat the case that we have $P_{m+1}^{\tilde t}\ge \lambda=\max\{\lambda_\mathrm{start}b_1,g_{b_1}\left(P_h\right)\}.$

Now we can derive the following contradiction:
\[P_{m+1}^{\tilde t} \ge g_{b_1}\left(P_h\right) \ge g_{b_1}\left(p\right)\ge g_{b_1}\left(f_{b_1}\left(P_{m+1}^{\tilde t}\right)\right)=g\left(f\left(\frac{P_{m+1}^{\tilde t}}{b_1}\right)\right)b_1>P_{m+1}^{\tilde t}.\]
For the second inequality, we use the monotonicity of $g_{b_1}(-)$. The third inequality follows from \Cref{le.2} and the last one from \Cref{pro.1.1}.
\end{proof}

\subsubsection{Establishing Main Lemma~\ref{ml:2}:}
Let $m\geq m(\varepsilon)$ be sufficiently large. The machine number~$m(\varepsilon)$ is determined by the proofs of \Cref{pro.endflat}
and \Cref{le.tech.Dts}, and then carries over to the subsequent lemmas. Let us assume for a contradiction sake that there was a stable sequence $\CJ$ such that
$\mathit{ALG}(\CJ) > (c+\varepsilon) OPT(\CJ)$. As argued in the beginning of Section~\ref{sec:adv}, see (\ref{eq:contr}), it suffices to show that
$P_{m+1}\ge \lambda_\mathrm{end}b_0$. If this was not the case, we would have $P_h\ge (1-\varepsilon)b_0$ by \Cref{le.P_h}.
In particular by Proposition~\ref{pro.1.2} we had $g_{b_0}\left(P_h\right)=g(1-\varepsilon)b_0>\lambda_\mathrm{end}b_0.$
But now \Cref{le.fill2} shows that
$P_{m+1}\ge \max \{\lambda_\mathrm{start}b_0,\min\{g_{b_0}\left(P_h\right),\lambda_\mathrm{end}b_0\}\}= \lambda_\mathrm{end}b_0$.

We conclude, by Corollary~\ref{cor:adv}, that $\mathit{ALG}$ is nearly $c$-competitive.

\section{Lower bounds}
We present lower bounds on the competitive ratio of any deterministic
online algorithm in the random-order model. Theorem~\ref{te.lb.2}
implies that if a deterministic online algorithm is $c$-competitive with
high probability as $m\rightarrow \infty$, then $c\geq 3/2$.

\begin{theorem}\label{te.lb.1}
Let $A$ be a deterministic online algorithm that is $c$-competitive in
the random-order model. Then $c\geq 4/3$ if $m\geq 8$.
\end{theorem}
\begin{theorem}\label{te.lb.2}
Let $A$ be a deterministic online algorithm that is nearly
$c$-competitive. Then $c\geq 3/2$.
\end{theorem}

A basic family of inputs are job sequences that consist of jobs having
an identical processing time of, say,~1. We first analyze them and then
use the insight to derive our lower bounds. Let $m\geq 2$ be arbitrary.
For any deterministic online algorithm $A$, let $r(A,m)$ be the maximum
number
in $\BN\cup\{\infty\}$ such that $A$ handles a sequence consisting of
$r(A,m)\cdot m$ jobs with an identical processing time of~$1$ by scheduling
each job on a least loaded machine.

\begin{lemma}\label{le.reg}
Let $m\geq 2$ be arbitrary. For every deterministic online algorithm
$A$, there exists a job sequence ${\cal J}$ such that
$A^{\rm rom}({\cal J}) \geq (1+\frac{1}{r(A,m)+1}) OPT({\cal J})$. We
use the convention that $\frac{1}{\infty+1}=0$.
\end{lemma}

\begin{proof}
For $r(A,m)=\infty$ there is nothing to show. For $r(A)<\infty$,
consider the sequence ${\cal J}$ consisting of $(r(A,m)+1)\cdot m$
identical jobs, each
having a processing time of~$1$. It suffices to analyze the algorithm
adversarially as all permutations of the job sequence are identical.
After having handled the first $r(A,m)\cdot m$ jobs, the algorithm $A$
has a schedule in which every machine has load of $r(A,m)$. By the
maximality
of $r(A,m)$, the algorithm $A$ schedules one of the following $m$ jobs on a machine
that is not a least loaded one. The resulting makespan is $r(A,m)+2$.
The lemma
follows since the optimal makespan is $r(A,m)+1$.
\end{proof}

\begin{proof}[Proof of \Cref{te.lb.1}]
Let $m\geq 8$ be arbitrary. Consider any deterministic online algorithm
$A$. If $r(A,m)\le 2$, then, by \Cref{le.reg}, there exists a sequence
${\cal J}$
such that $A^{\rm rom}({\cal J}) \geq {4\over 3} \cdot OPT({\cal J})$.
Therefore, we may assume that $r(A,m)\ge 3$. Consider the input sequence
$\CJ$ consisting of $4m-4$ identical small jobs of processing time~$1$ and
one large job of processing time~$4$. Obviously $\OPT(\CJ)=4$.

Let $i$ be the number of small jobs preceding the large job in
$\CJ^\sigma$. The random variable $i$ takes any (integer) value between
$0$ and $4m-4$
with probability $\frac{1}{4m-3}$. Since $r(A,m)\ge 3$ the least loaded
machine has load of at least $l=\left\lfloor \frac{i}{m}\right\rfloor$ when
the large job arrives. Thus $A(\CJ^\sigma)\ge l+4$. The load $l$ takes
the values $0$, $1$ and $2$ with probability $\frac{m}{4m-3}$ and the value
$3$ with probability $\frac{m-3}{4m-3}$. Hence the expected makespan of
algorithm $A$ is at least
\[A^\mathrm{rom}(\CJ)\ge\frac{m}{4m-3}\cdot(0+1+2) + \frac{m-3}{4m-3}
\cdot 3 +4=\frac{6m-9}{4m-3} +4  > \frac{16}{3}=\frac{4}{3}\OPT(\CJ).\]
For the last inequality we use that $m\ge 8$.
\end{proof}

\begin{proof}[Proof of \Cref{te.lb.2}]
Let $m\geq 2$ be arbitrary and let $A$ be any deterministic online
algorithm. If $r(A,m)=0$, then consider the sequence ${\cal J}$
consisting of $m$ jobs with a processing time of~1 each. On every
permutation of ${\cal J}$ algorithm $A$ has a makespan of~2, while
the optimum makespan is~1. If $r(A,m)\geq 1$, then consider the sequence
${\cal J}$ consisting of $2m-2$ small jobs having a processing time
of~$1$ and one large job with a processing time of~$2$. Obviously
$OPT(\CJ)=2$. If the permuted sequence starts with $m$ small jobs, the
least
loaded machine has load $1$ once the large job arrives. Under such
permutations $A(\CJ^\sigma)\ge3=\frac{3}{2}\cdot \OPT(\CJ)$ holds true.
The probability of this happening is $\frac{m-1}{2m-1}$. The probability
approaches~$\frac{1}{2}$ and in particular does not vanish,
for $m\rightarrow\infty$. Thus, if $A$ is nearly $c$-competitive, then
$c\geq 3/2$.
\end{proof}

\bibliographystyle{plain}
\let\oldbibliography\thebibliography
\renewcommand{\thebibliography}[1]{%
  \oldbibliography{#1}%
  \setlength{\itemsep}{0pt plus .3pt}
  \setlength{\parsep}{0pt plus .3pt}
   \setlength{\parskip}{0pt plus .3pt}
}

\bibliography{icalp20}


\section*{Appendix}\label{app}
\llI*

\begin{proof}
Fix any proper job sequence $\CJ$. For any ${\cal J}^\sigma$, let $N(\varphi+\varepsilon)[\sigma] = N(\varphi+\varepsilon)[{\cal J}^\sigma]$.
Furthermore, let $\tilde N \left(\varphi+\frac{\varepsilon}{2}\right)[\sigma]$ denote the number of the $m+1$ largest jobs of ${\cal J}$
that appear among the first $\left\lfloor\left(\varphi+\frac{\varepsilon}{2}\right)n\right\rfloor$ jobs in $\CJ^\sigma$.
Then we derive by the inclusion-exclusion principle:
\begin{align*}&\bP_{\sigma\sim S_n}\left[ N(\varphi+\varepsilon)[\sigma] \ge \lfloor \varphi m\rfloor+h+2 \right]\\
\ge &\bP_{\sigma\sim S_n}\left[ \tilde N\left(\varphi+\frac{\varepsilon}{2}\right)[\sigma] \ge \lfloor \varphi m\rfloor+h+2 \textrm{ and } L^{\left\lfloor \left(\varphi+\frac{\varepsilon}{2}\right) n\right\rfloor}[\sigma]< (\varphi+\varepsilon)L\right]\\
\ge &\bP_{\sigma\sim S_n}\left[ \tilde N\left(\varphi+\frac{\varepsilon}{2}\right)[\sigma] \ge \lfloor \varphi m\rfloor+h+2\right]
+\bP_{\sigma\sim S_n}\left[ L^{\left\lfloor \left(\varphi+\frac{\varepsilon}{2}\right) n\right\rfloor}[\sigma]< (\varphi+\varepsilon)L\right]-1.\end{align*}
By the Load Lemma the second summand can be lower bounded for every proper sequence $\CJ$ by a term approaching $1$ as $m\rightarrow\infty$.
Hence it suffices to verify that this is also possible for the term
\[\bP_{\sigma\sim S_n}\left[ \tilde N\left(\varphi+\frac{\varepsilon}{2}\right)[\sigma] \ge \lfloor \varphi m\rfloor+h+2\right].\]
We will upper bound the probability of the opposite event by a term approaching $0$ for $m\rightarrow\infty$. The random variable $\tilde N\left(\varphi+\frac{\varepsilon}{2}\right)[\sigma]$ is hypergeometrically distributed and therefore has expected value
\[E = \frac{\left\lfloor\left(\varphi+\frac{\varepsilon}{2}\right)n\right\rfloor}{n}(m+1) \ge \left(\varphi+\frac{2}{5}\varepsilon\right)(m+1) .\]
Recall that for proper sequences $n>m$ holds. For the second inequality we require $m$ and hence in also~$n$ to be large enough such that $\frac{1}{n} \le \frac{\varepsilon}{10}$ holds. Again, the variable $\tilde N\left(\varphi+\frac{\varepsilon}{2}\right)[\sigma]$ is hypergeometrically distributed and its variance is thus
\[V=\frac{\left\lfloor\left(\varphi+\frac{\varepsilon}{2}\right)n\right\rfloor\left(n-\left\lfloor\left(\varphi+\frac{\varepsilon}{2}\right)n\right\rfloor\right)(m+1)(n-m-1)}{n^2(n-1)}\le m+1.\]
Note that we have for $m$ large enough:
\[\lfloor \varphi m\rfloor+h+2\le \left(1+\frac{1}{5}\varepsilon\right)\varphi (m+1) \le E - \frac{\varepsilon\varphi \sqrt{m+1}}{5}\sqrt V .\]
Hence, using Chebyshev's inequality, we have
\begin{align*}&\bP_{\sigma\sim S_n}\left[\tilde N\left(\varphi+\frac{\varepsilon}{2}\right)[\sigma] < \lfloor \varphi m\rfloor+h+2\right]\\
\le &\bP_{\sigma\sim S_n}\left[E- \tilde N\left(\varphi+\frac{\varepsilon}{2}\right)[\sigma] > \frac{\varepsilon\varphi \sqrt{m+1}}{5}\sqrt V \right]
\\ \le &\frac{25}{\varepsilon^2\varphi^2(m+1)}\end{align*}
and this term vanishes as $m\rightarrow\infty$.
\end{proof}

\letechDts*

\begin{proof}
Let $t=t(s)-1$.
We have $l_j^{t}\ge b$ for $j\le s-1$ as the first $s-1$ machines are full. Considering the load on the machines numbered up to $m-h$ we obtain
\begin{align*}L_{i+1}^{t}&\ge\frac{\sum_{j=i+1}^{s-1}l_j^{t}+(m-h-s+1)\overline{L}_s^{t}}{m-i}\\
&\ge \frac{(s-i-1)b+(m-h-s+1)\alpha^{-1}b}{m-i}\\
&\ge \alpha^{-1}b+\frac{s-i-1}{m-i}(1-\alpha^{-1})b - {h\over m-i} \alpha^{-1}b\\
&= \alpha^{-1}b+\frac{s-i-1}{m-i}\frac{b}{2(c-1)} - {h\over m-i} \alpha^{-1}b.
\end{align*}
If $s>i+1$, the schedule was steep at time $t=t(s)-1$ and hence
\[
l_k^{t} \ge \alpha L_{i+1}^{t} 
> b + \frac{s-i-1}{m-i}\frac{\alpha  \cdot b}{2(c-1)} - {h\cdot b\over m-i}.
\]
Since $l_k^{t}\ge l_i^{t} \ge b$, the previous inequality holds for $s=i+1$, too, no matter whether $J^{(s)}=J^{(i+1)}$ was scheduled flatly or steeply. We hence get, for all $s\ge i+1$,
\begin{align*}
L^t &\ge \frac{kl_k^t+(s-k-1)l_{s-1}^t+(m-h-s+1)\overline{L}^t_s}{m}\\
&> \frac{k\left(b + \frac{s-i-1}{m-i} \cdot \frac{\alpha \cdot b}{2(c-1)}-{h\cdot b\over m-i}\right)+(s-k-1)b+(m-h-s+1)\left(b-\frac{b}{2(c-1)}\right)}{m}\\
&= \left(1+\frac{1}{2(c-1)}\left(\frac{k}{m}\frac{s-i-1}{m-i} \cdot\alpha-\frac{m-s+1}{m}\right)\right) b - 
\left({k\cdot h\over m(m-i)}+ {h\cdot \alpha^{-1}\over m}\right) b.
\end{align*}
In the above difference, we first examine the first term, which is minimized if $s=i+1$. With this setting it is still lower bounded
by \[\left(1-\frac{m-i}{2(c-1)m}\right)b > \left(1-\frac{2(2-c)}{2(c-1)}\right)b\approx 0.8205 \cdot b >\frac{3b}{4}.\]
In the second term of the above difference ${k \over m-i} = {2i-m\over m-i}$ is increasing in $i$, where $i \leq (2c-3)m+h+1$. 
We choose $m$ large enough such that 
$${k\over m-i} \leq {(4c-7)m+2(h+1) \over 2(2-c)m-(h+1)} \leq 1.5.$$ 
There holds $\alpha^{-1} < 0.5$. Thus the second term in the difference is upper bounded by ${2hb\over m}$.

Recall that $b>(c-1+\varepsilon)OPT$. Furthermore, $0<\varepsilon <2-c$ such that $c-1+\varepsilon <1$. Therefore, we obtain
\begin{align*}
L^t &> \left(c-1+\frac{1}{2}\left(\frac{\textbf{k}}{m}\frac{s-i-1}{m-\textbf{i}} \cdot\alpha-\frac{m-\textbf{i}}{m}+\frac{s-i+1}{m}\right)
+\frac{3}{4}\varepsilon -{2h\over m}\right)OPT.
\end{align*}
In the previous term we intentionally highlighted three variables. It is easy to check that if we decrease these variables, the term decreases, too. We do this by setting $\textbf{k}=(4c-7)m$ and $\textbf{i}=(2c-3)m$ (while ignoring the non-highlighted occurrences of $i$). We also assume that $m$ is large enough such that $\frac{\varepsilon}{4}\ge \frac{3h+2}{m}$. Then the previous lower bound on $L^t$ can be brought to the following form:

\begin{align*}
L^t &> \left(2c-3 + \frac{1}{2}\left(\frac{4c-7}{2(2-c)}\cdot\alpha+1\right)\frac{s-i-1}{m}+\frac{h+2}{m}+\frac{\varepsilon}{2}\right)OPT\end{align*}
Using that $\frac{i+1}{m}< 2c-3+\frac{h+2}{m}$ and evaluating the term in front of $\frac{s-i-1}{m}$ we get
\[ L^t >\left(\frac{i+1}{m}+1.0666\cdot\frac{s-i-1}{m}+\frac{\varepsilon}{2}\right)OPT
> \left(\frac{s}{m}+\frac{\varepsilon}{2}\right)OPT. \]
The lemma follows by noting that $OPT \ge L$.
\end{proof}

\end{document}